\documentclass{article}

\usepackage{arxiv}

\usepackage[utf8]{inputenc}
\usepackage[T1]{fontenc}
\usepackage{enumitem}
\usepackage{url}
\usepackage{listingsutf8}
\usepackage{comment}
\usepackage{nicefrac}      
\usepackage{lipsum}

\usepackage[english]{babel}

\usepackage{mathtools,amsfonts,amssymb,amsthm,mathrsfs}
\usepackage{empheq}
\usepackage{multirow}
\usepackage{changepage}
\usepackage{cases}
\usepackage[separate-uncertainty = true, multi-part-units=repeat]{siunitx}

\usepackage{graphicx} 
\usepackage{float}
\usepackage{caption}
\usepackage{subcaption}
\usepackage[export]{adjustbox}

\usepackage{hyperref}
\hypersetup{pdfauthor=author}
\usepackage{cite}
\usepackage{booktabs}
\usepackage{orcidlink}

\newtheorem{prop}{Proposition}
\newtheorem{definition}{Definition}
\newtheorem{remark}{Remark}
\newtheorem{lemma}{Lemma}
\newtheorem{corollary}{Corollary}

\DeclareSIUnit\mic{MIC}
\DeclareSIUnit\msc{MSC}
\DeclareSIUnit\mdk{MDK_{99}}
\DeclareSIUnit\mol{M}
\DeclareSIUnit\cfus{CFU}
\DeclareSIUnit\ub{ub}
\DeclareSIUnit\us{us}

\DeclarePairedDelimiter\ceil{\lceil}{\rceil}

\newcommand{\nn}

\title{Conditions for Bacterial Selection and Extinction Driven by Growth–Kill Trade‑Off in Cyclic Antimicrobial Treatments}

\author{
 $^\ast$Nerea Martínez-López \orcidlink{0000-0002-3976-9883}\\
  Biosystems \& Bioprocesses Engineering\\
  Spanish National Research Council (IIM-CSIC)\\
  Vigo, Spain, 36208 \\
  \texttt{$^\ast$nmartinez@iim.csic.es} \\
   \And
 Niclas Nordholt \orcidlink{0000-0001-5788-0801}\\
  Division of Biodeterioration and Reference Organisms\\
  Federal Institute for Materials Research and Testing (BAM)\\
  Berlín, Germany, 12205\\
   \And
 Frank Schreiber \orcidlink{0000-0003-1957-6328}\\
  Division of Biodeterioration and Reference Organisms\\
  Federal Institute for Materials Research and Testing (BAM)\\
  Berlín, Germany, 12205\\
  \And
  Míriam R. García \orcidlink{0000-0001-9538-6388}\\
  Biosystems \& Bioprocesses Engineering\\
  Spanish National Research Council (IIM-CSIC)\\
  Vigo, Spain, 36208 \\
}

\begin{document}


\maketitle


\begin{abstract}
Antimicrobial protocols — using substances such as antibiotics or disinfectants — remain the preferred option for preventing the spread of pathogenic bacteria. However, bacteria can develop mechanisms to reduce their antimicrobial susceptibility, which can lead to treatment failure and the selection of resistance or tolerance. In this work, we propose a minimal population dynamics model to study bacterial selection during cyclic antimicrobial application, a commonly used protocol. Selection in bacterial populations with heterogeneous antimicrobial susceptibility is modelled here as a trade-off between survival advantage (reduction in antimicrobial killing) and potential fitness costs (reduction in growth rate) of the less susceptible strains. The proposed model allows us to derive useful expressions for determining the success of cyclic antimicrobial treatments based on two bacterial traits: growth and kill rates. The results obtained here are directly applicable to preventing the selection and spread of resistant and tolerant bacterial strains in real-life protocols.
\end{abstract}



\section{Introduction}\label{sec:Introduction}

Antimicrobial protocols are used daily to eliminate bacteria and prevent the spread of pathogens in home~\cite{2016Wesgate}, industrial or clinical~ settings. However, bacteria from the same population can respond to treatment very differently, as a result of genetic or phenotypic changes altering antimicrobial susceptibility, among other bacterial characteristics. That is the case of antimicrobial resistance (AMR) or tolerance, which can seriously compromise antimicrobial treatment. In these scenarios, even a small subpopulation of bacteria with reduced antimicrobial susceptibility can survive treatment, making evolutionary processes—mutation, selection, and competition—highly relevant for understanding treatment failure and development of AMR and tolerance~\cite{nordholt_persistence_2021}. Then, accurately predicting the dynamics of AMR and tolerance selection under biocidal stress is essential for systematic optimisation of antimicrobial protocols~\cite{acar_enhanced_2019}.

Although the evolutionary process is fundamentally stochastic, especially in small populations \cite{coates2018antibiotic}, outcomes become more predictable and often repeatable under strong and recurrent selective pressures, such as those imposed by cyclic biocide treatments \cite{charlebois_quantitative_2023}. Mutation itself is random; however, in practice only a few mutant lineages tend to expand and dominate the population even at antimicrobial concentrations far below the minimum inhibitory concentration~\cite{petrungaro_antibiotic_2021}. In other words, evolution is highly repeatable when the efficiency of selection is high relative to the stochasticity of genetic drift, mutation, recombination and unpredictable environmental changes. Consequently, most mathematical models of antimicrobial resistance (AMR) require deterministic equations~\cite{wortel_towards_2023}, typically systems of Ordinary Differential Equations (ODEs)~\cite{schardong_mapping_2026}.

If the process can be assumed to be deterministic under strong selective pressures, as is often observed in cyclic treatments or in treatments that are highly bactericidal, then it becomes necessary to identify a useful and experimentally applicable theory that will allow us to determine the conditions under which bacterial populations will go extinct during these treatment regimens. When a population is driven to extinction, it is also important to estimate the time to extinction so that unnecessary additional treatment cycles can be avoided. Conversely, when survival occurs, it becomes essential to characterise the conditions under which resistant or tolerant subpopulations gain a selective advantage, enabling treatment protocols to be adjusted to minimise their enrichment. 

The establishment of this theoretical framework constitutes the main unresolved challenge, and overcoming it is indispensable for the optimisation of effective disinfection or sanitation protocols \cite{wortel_towards_2023,2024Nordholt}. Recent research has derived extinction conditions for bacterial populations exposed to antimicrobials, although this work assumes a constant antimicrobial stress without cyclic or time‑varying concentration dynamics, and accounts for variability through stochastic fluctuations in growth rather than through the explicit coexistence of strains with different traits \cite{mansour_stochastic_2023}. Other studies use the birth–death stochastic process to investigate population dynamics under sequential therapies, capturing complex evolutionary phenomena such as collateral sensitivity but without providing explicit analytical conditions for selection or extinction \cite{anderson_invariant_2025}. As a consequence, these approaches require numerical solution through dedicated in silico tools that must be recalibrated for each specific scenario and that usually depended on large parameter sets \cite{anderson_computational_2025}. Collectively, these studies motivate simple, analytically tractable models of strain competition and extinction under cyclic antimicrobial killing, with minimal population dynamics models appearing most useful. However existing results provide selection conditions only for growth trait evolution under serial dilution without antimicrobial treatment \cite{2020Lin}.



In this work, we study selection during cyclic antimicrobial application for bacterial populations formed by multiple strains. First, we propose a mathematical model to describe the population dynamics under cyclic treatment, consisting of successive growth and antimicrobial-killing periods.  The model formulation is kept simple to study selection as a function of the minimum bacterial characteristics (traits) to describe the dynamics of the different strains, namely, growth rates during recovery and kill rates in response to antimicrobial killing. Then, the mathematical model is used to derive conditions for determining extinction or selection of the different strains based on the bacterial traits and the setup parameters for cyclic treatment. We also provide important measures to quantify selection and compare the selective ability of the strains, such as the extinction cycles and the selection coefficients. The resulting approach provides directly applicable criteria for preventing the selection and spread of resistant or tolerant bacteria in real-life antimicrobial protocols.

\section{The mathematical model}\label{sec:MathModel}

\subsection{Model assumptions}\label{sec:MathModel:subsec:Assumptions}

Consider a bacterial population formed by $n>1$ distinct strains $\{S_i\mid 1\leq i\leq n\}.$ The population is subjected to a cyclic treatment that alternates growth periods with killing periods via some bactericidal antimicrobial, e.g., an antibiotic or disinfectant. The duration of the growth periods $(t_g>0)$ and the killing periods $(t_k>0)$ remains constant between cycles, as well as the applied antimicrobial concentration $(C\geq 0).$ We provide, though, an extension of the model 
to the case where the antimicrobial dose varies across cycles. 
Dilution into fresh medium  by a factor $0<D<1$ follows each killing period to remove the antimicrobial and provide the necessary resources for propagating bacteria to the next cycle. 

During the growth periods, the strains grow exponentially until the total population reaches the saturation concentration (or carrying capacity $K>0),$ once the resources are exhausted. When no more resources are available, bacteria enter the stationary phase and can no longer reproduce. Then, the strains indirectly compete  for the available resources during growth;  the fastest-growing strains will have a competitive advantage because they can convert resources into biomass more rapidly before saturation. Conversely, the slower-growing strains will be at a competitive disadvantage, as they may not be able to reproduce significantly before resources are exhausted. 
Thus, we will consider the exponential growth rates $\{\mu_i>0\mid 1\leq i\leq n\}$ as the main parameters characterising the population dynamics during growth. For simplicity, we ignore the influence of other growth-related parameters, such as bacterial lag or growth efficiency (proportion of the resource converted by bacteria into biomass). 

The growth periods are followed by killing periods, during which we assume the strains' concentration decay exponentially  (first-order kinetics). The kill rates $\{k_i(C)\geq 0\mid 1\leq i\leq n\}$  are considered the main parameters characterising the population dynamics during antimicrobial killing. We thus ignore bimodal killing kinetics arising from phenotypic heterogeneity in the strains, such as heteroresistance \cite{martinez-lopez_moment_2024_bis} or persistence \cite{nordholt_persistence_2021}. On the other hand, we assume that the strains cannot growth during the killing periods 

To confer the model more realism, we assume the strains cannot reproduce below an extinction limit $X_e>0.$ We will generally consider that the initial cell concentration (inoculum) for the strains, $\{X_{i0}\mid 1\leq i\leq n\},$ verify $X_e\leq X_{i0},$  so that they can reproduce during the first growth period, unless stated otherwise (e.g., when studying the dynamics for the strains in isolation). 

The duration of the growth and killing  periods, the carrying capacity, the strains' inocula, the dilution factor, the applied antimicrobial concentration, and the extinction limit, will be called the setup parameters for cyclic treatment. On the other hand, the growth  and kill rates are the main parameters characterising the population dynamics, so that we will refer to them as the bacterial traits. Note that we can assume the carrying capacity is a setup parameter (adjusted through the initial resource concentration in the culture medium) because we ignore differences in growth efficiency between the strains. 

\subsection{Model formulation}\label{sec:MathModel:subsec:Formulation}

Let $\{X_i(t)\mid 1\leq i\leq n\}$ denote the cell concentrations of the different strains in the population at any time $t\geq 0.$ Taking into account the hypothesis and parameters described in~\nameref{sec:MathModel:subsec:Assumptions}, the following Ordinary Differential Equations (ODEs) describe the dynamics of the bacterial population under cyclic antimicrobial treatment:
\begin{align}\label{eq:ODEs_simple_model}
 \dfrac{dX_i}{dt}(t) &= \left(\mu_i  I_g(t)I_{e}\left(X_{i,c}^0\right) - k_i(C)  I_k(t)\right)X_i(t),\quad t_{c}^0 < t < t_{c}^f,\quad c\geq 1, \quad 1\leq i\leq n.
\end{align}

The duration of the growth and killing periods remain fixed, so that $t_c^0=(c-1)\left(t_g+t_k\right)$ and $t_c^f=c\left(t_g+t_k\right)$ are the initial and final time for the $c$-cycle, respectively. Note that the final time of the $c$-cycle coincides with the initial time of $c+1,$ i.e., $t_{c}^f=t_{c+1}^0.$  We denote $X_{i,c}^0 =X_i(t_c^0)$ the initial condition for ODEs~\eqref{eq:ODEs_simple_model} at the $c$-cycle, which is:
\begin{align}\label{eq:IC_ODEs_simple_model}
   X_{i,c}^0&=
    \begin{cases}
    X_{i0},& c=1,\\
    D X_i(t_{c-1}^{f-}),& c>1,
    \end{cases}\quad X_i(t_{c-1}^{f-})=\lim_{t\rightarrow t_{c-1}^{f-}} X_i(t),\quad 1\leq i\leq n.
\end{align}

The functions $I_g,$ $I_e$ and $I_k$ involved in ODEs~\eqref{eq:ODEs_simple_model} are indicators for exponential growth, strains' extinction and antimicrobial killing, respectively. In first place, the function $I_g$ indicates saturation in the total cell concentration due to resources exhaustion, so that:
\begin{align}\label{fun:I_g}
    I_g(t) = \begin{cases}
        1,& t \in \left(t_c^0, t_c^0 + t_c\right)\textmd{ for some }c\geq 1,\\
        0,&\textmd{i.o.c.},
    \end{cases}\qquad t_c = \min\{t_g, t_c^{sat}\},
\end{align}
where $t_c$ denotes the time spent by the strains in the exponential growth phase at the $c$-cycle, and  $t_c^{sat}$ is the saturation time. That is, $t_c^{sat}$ denotes the time necessary for the population to reach the carrying capacity at the $c$-cycle. Note that the saturation time for the $c$-cycle is the unique solution (if any) of the following equation:
\begin{align}\label{eq:t_sat}
       \sum_{i=1}^n X_{i,c}^0\exp\left(\mu_{i}t_c^{sat}I_e(X_{i,c}^0)\right)=K,\quad c\geq 1,
\end{align}
which can only be solved numerically in the most general case. Nevertheless, as we will see, the model analysis can be fairly developed without using any closed-form expression for the saturation time in Eq.~\eqref{eq:t_sat}. An approximation is provided, though, for cases where this is necessary (see~ \nameref{sec:MathModel:subsec:Approx_tsat}). 

On the other hand, the function $I_e$ in ODEs~\eqref{eq:ODEs_simple_model} accounts for strains' extinction, so that they cannot reproduce during the $c$-cycle if their concentration is below the extinction limit at the start of the cycle. That is:
\begin{align}\label{fun:I_E}
    I_e\left(X_{i,c}^0\right) = \begin{cases}
        1,& X_{i,c}^0 \geq X_e,\\
        0,&\textmd{i.o.c.,}
    \end{cases}\quad c\geq 1,\quad 1\leq i\leq n.
\end{align}

Finally, the function $I_k$ in ODEs~\eqref{eq:ODEs_simple_model} indicates the periods of antimicrobial killing:
\begin{align}\label{fun:I_k}
    I_k(t) = \begin{cases}
        1,& t \in \left(t_c^0+t_g, t_c^f\right)\textmd{ for some }  c\geq 1,\\
        0,&\textmd{i.o.c.}
    \end{cases}
\end{align}

\begin{remark}\label{rmk:sol_tsat}
    If Eq.~\eqref{eq:t_sat} has a solution, it is unique. On the other hand, the solution to Eq.~\eqref{eq:t_sat} exists if and only if the cell concentration of some strain is above the extinction limit at the start of the $c$-cycle. Otherwise, the population cannot grow, nor therefore reach saturation. Note also that the unique solution of Eq.~\eqref{eq:t_sat} verifies $t_c^{sat}\leq 0$ if the total cell concentration at the start of the $c$-cycle is such that:
    $$\sum_{i=1}^n X_{i,c}^0 \geq K.$$ 
    
    In that case, the duration of the exponential growth phase at the $c$-cycle would be $t_c=0.$ Nevertheless, this scenario can occur only when the total cell concentration at the start of treatment exceeds the carrying capacity, since the population is diluted at the end of each cycle (see Eq.~\eqref{eq:IC_ODEs_simple_model}).
\end{remark}

\begin{remark}\label{rem:limXf}
    The cell concentrations $\{X_i(t)\mid 1\leq i\leq n\}$ are considered discontinuous at $t_c^f=t_{c+1}^0$ for any cycle $c\geq 1,$ since Eq.~\eqref{eq:IC_ODEs_simple_model} considers an instantaneous dilution step at the end of each cycle. That is:
    \begin{align*}
     \lim_{t\rightarrow t_c^{f^+}} X_i(t) = X_{i,c+1}^0 = D \lim_{t\rightarrow t_c^{f-}} X_i(t)\underbrace{\neq}_{0<D<1} \lim_{t\rightarrow t_c^{f-}} X_i(t),\quad c\geq 1,\quad 1\leq i\leq n.
    \end{align*} 
    
    Note that Eqs.\eqref{eq:ODEs_simple_model}--\eqref{eq:IC_ODEs_simple_model} imply that the cell concentrations at the final time $t_c^f$ are:
    \begin{align*}
        X_i(t_c^f) = X_i(t_{c+1}^0)=X_{i,c+1}^0,\quad c\geq 1,\quad 1\leq i\leq n.
    \end{align*}
\end{remark}

The model for cyclic antimicrobial treatment given in Eqs.~\eqref{eq:ODEs_simple_model}--\eqref{fun:I_k} has the following (unique) analytical solution for the cell concentrations $\{X_i(t)\mid 1\leq i\leq n\}$ of the different strains:
\begin{align}\label{eq:sol_ODEs}
    X_i(t) = 
    X_{i,c}^0\exp\left(\mu_i\min\{t_c,t - t_c^0\}I_e\left(X_{i,c}^0\right) - k_i(C)\max\{0,t-(t_c^0 +  t_g)\}\right),\;\; t_c^0 \leq t < t_c^f,\;\; c\geq 1,
\end{align}
where $X_{i,c}^0$ and $I_e\left(X_{i,c}^0\right)$ are calculated recursively using Eq.~\eqref{eq:IC_ODEs_simple_model} and Eq.~\eqref{fun:I_E}, respectively. 

The analytic solution in Eq.~\eqref{eq:sol_ODEs} is rather cumbersome and will not be applied directly for the analysis and demonstrations in this work. Instead, we will repeatedly use the cell concentrations at the end of each cycle, which are:
\begin{align}\label{eq:Xf}
    X_i\left(t_c^f\right) = 
D^cX_{i0}\exp\left(\mu_i\sum_{q=1}^{c}t_qI_e\left(X_{i,q}^0\right) - ck_i(C)t_k\right),\quad c\geq 1,\quad 1\leq i\leq n,
\end{align}
and the following expression, relating the cell concentrations at the initial and final time of each cycle:
\begin{align}\label{eq:Xf_fun_X0}
    X_i\left(t_c^f\right) = 
        DX_{i,c}^0\exp\left(\mu_it_cI_e\left(X_{i,c}^0\right) - k_i(C)t_k\right),\quad c\geq 1,\quad 1\leq i\leq n.
\end{align}

\begin{remark}
The model in Eqs.~\eqref{eq:ODEs_simple_model}--\eqref{fun:I_k} can be easily extended to cyclic treatments for which the antimicrobial dose varies between cycles. In that case,  ODEs~\eqref{eq:ODEs_simple_model} transform to:
\begin{align}\label{eq:ODEs_simple_model_varC}
\dfrac{dX_i}{dt}(t) &= \left(\mu_i  I_g(t)I_{e}\left(X_{i,c}^0\right) - k_i(C_c)  I_k(t)\right)X_i(t),\quad t_{c}^0 < t < t_{c}^f,\quad  1\leq i\leq n,\quad c\geq 1,
\end{align}
where $C_c\geq 0$ is the antimicrobial concentration applied at the $c$-cycle, and Eqs.~\eqref{eq:IC_ODEs_simple_model}--\eqref{fun:I_k} remain unchanged. Then, the kill rates for the different strains at the applied antimicrobial doses should be considered in the model as independent parameters. One alternative is to consider some pharmacodynamic function relating the kill rates of the strains with the antimicrobial concentration. For example, we can use the following Hill-type dependence:
\begin{align*}
    k_i(C) = k_{max,i} \frac{C^{H_i}}{C^{H_i} + EC_{50,i}^{H_i}},\quad 1\leq i\leq n,
\end{align*}
where $k_{max,i}$ is the maximum kill rate, $EC_{50,i}$ is the half-maximum inhibitory concentration, and $H_i$ is the Hill exponent. In this case, only three parameters are needed per strain to consider antimicrobial killing at any concentration.

Note that the extended model in Eqs.~\eqref{eq:IC_ODEs_simple_model}--\eqref{fun:I_k} and \eqref{eq:ODEs_simple_model_varC} for cyclic treatment with variable antimicrobial concentration has a unique analytical solution, given by:
\begin{align}
    X_i(t) = 
    X_{i,c}^0\exp\left(\mu_i\min\{t_c,t - t_c^0\}I_e\left(X_{i,c}^0\right) - k_i(C_c)\max\{0,t-(t_c^0 +  t_g)\}\right),\;\; t_c^0 \leq t < t_c^f,\;\; c\geq 1.
\end{align}

However, the analysis of the extinction and selection dynamics for the extended model is complicated and will not be addressed here. Note that little can be said about the extinction or selection of the strains if the succession of applied antimicrobial doses $\{C_c\mid c\geq 1\}$ does not take a particular form (e.g., increases or decreases with the cycles).
\end{remark}

\subsection{An approximation for the saturation time}\label{sec:MathModel:subsec:Approx_tsat}

The model for cyclic antimicrobial treatment given by Eqs.~\eqref{eq:ODEs_simple_model}--\eqref{fun:I_k} has an analytical solution for the cell concentrations of the different strains in the population. However, these depend on the saturation times satisfying Eq.~\eqref{eq:t_sat}, which has no closed-form solution. Then, the saturation times can be numerically approximated in model simulations, but there is no closed-form expression relating the saturation times to the model parameters. 

Therefore, in this section, we derive approximate expressions for the saturation times of the bacterial population during cyclic treatment. We start by noting that, summing ODEs~\eqref{eq:ODEs_simple_model}, the total cell concentration during the exponential growth phase verifies, at any cycle $c\geq 1,$  the following initial value problem:
\begin{align*}
    \begin{cases} 
    \dfrac{dX}{dt}(t)=\sum_{i=1}^n \mu_i I_e(X_{i,c}^0) X_i(t),\\[0.1cm]
    \quad\:\: X_c^0 = \sum_{i=1}^n X_{i,c}^0,
   \end{cases}\quad t_c^0 < t < t_c^0 + t_c,\quad c\geq 1,
\end{align*}
where the total cell concentration is denoted as:
\begin{align*}
    X(t)=\sum_{i=1}^n X_i(t).
\end{align*}

The next step is to rewrite the above ODE as follows:
\begin{align*}
    \frac{dX}{dt}(t)= \left(\sum_{i=1}^n \mu_i I_e(X_{i,c}^0)  \frac{X_i(t)}{X(t)}\right) X(t) = \mu(t) X(t),\quad t_c^0 < t < t_c^0 + t_c,\quad c\geq 1,
\end{align*}
so that $\mu(t)$ is the exponential growth rate of the bacterial population, depending on the (time-varying) subpopulation frequencies of the different strains. That is, we define the growth rate of the total population as:
\begin{align*}
    \mu(t) = \sum_{i=1}^n \mu_iI_e(X_{i,c}^0)  \frac{X_i(t)}{X(t)},\quad t_c^0 < t < t_c^0 + t_c,\quad c\geq 1.
\end{align*}

Note that we can divide the previous ODE by the total cell concentration since, taking into account Eq.~\eqref{eq:sol_ODEs}, we have $X(t)>0$ at any time $t\geq 0,$ as long as $X_{i0}>0$ for some strain $S_i$ $(1\leq i\leq n).$ Additionally, the cell concentration of some strain $S_i$ should be higher than or equal to the extinction limit at the start of the $c$-cycle (i.e., $X_{i,c}^0\geq X_e),$ or it makes no sense to talk about the saturation time (see Rmk.~\ref{rmk:sol_tsat}).

We now assume that the growth rate $\mu(t)$ of the population during the $c$-cycle can be approximated by a constant value $\hat\mu_c,$ depending on the cell concentration of the different strains at the start of  growth. The proposed approximation is:
\begin{align*}
    \mu(t) \approx \hat \mu_c = \sum_{i=1}^n \mu_i I_e(X_{i,c}^0) \cfrac{X_{i,c}^0}{X_c^0},\quad t_c^0 < t < t_c^0 + t_c,\quad c\geq 1,
\end{align*}
and, consequently, the solution to the ODE describing the total cell concentration dynamics during exponential growth can be approximated as:
\begin{align*}
 X\left(t\right)\approx X_c^0\exp\left(\hat \mu_c t\right),\quad t_c^0 < t < t_c^0 + t_c,\quad c\geq 1,
\end{align*}

Then, we can finally obtain the desired approximation for the saturation time by solving the following equation:
\begin{align*}
X\left(t_c^0+t_c^{sat}\right)\approx X_c^0\exp\left(\hat \mu_c t_c^{sat}\right)=K,\quad c\geq 1,
\end{align*}
whose solution is:
\begin{align}\label{eq:t_sat_approx}
    t_c^{sat} \approx \hat t_c^{sat}=\frac{X_{c}^0}{\sum_{i=1}^n \mu_i I_e(X_{i,c}^0) X_{i,c}^0}\log\left(\frac{K}{X_{c}^0}\right),\quad c\geq 1.
\end{align}

\section{Extinction and selection conditions}\label{sec:Sel&Ext_Cond}

The mathematical model in Eqs.~\eqref{eq:ODEs_simple_model}--\eqref{fun:I_k}  is useful for determining the outcome of cyclic antimicrobial treatment, e.g., which strains will eventually be selected, which will go extinct, or the conditions favouring the selection of one strain over the others. That is particularly interesting when the bacterial population consists of an ancestor strain and several resistant or tolerant mutants, which show reduced antimicrobial susceptibility at the expense of a fitness cost (typically reduced growth)~\cite{2015Melnyk,nordholt_persistence_2021,2024Nordholt}. Then, the trade-off between increased survival and fitness cost of the evolved mutants determines the selection dynamics under cyclic treatment and, in particular, whether the mutants outcompete the ancestor. 

Throughout this section, we use Eqs.~\eqref{eq:ODEs_simple_model}--\eqref{fun:I_k} to derive simple expressions for determining the extinction or selection of the different strains in the long term (i.e., when $c\rightarrow \infty)$ depending on the bacterial traits and the setup parameters for the cyclic protocol. First, let us define the concepts of strains' survival and extinction used here.

\begin{definition}\label{def:ext}
Consider a bacterial population formed by strains $\{S_i\mid 1\leq i\leq n\}$ subjected to cyclic antimicrobial treatment, as described by Eqs.~\eqref{eq:ODEs_simple_model}--\eqref{fun:I_k}. 
We say that:
\begin{enumerate}[label=\arabic*)]
    \item  Strain $S_i$ goes extinct if $\lim_{c\rightarrow \infty} X_i\left(t_c^f\right) =0.$ 
    \item Strain $S_i$ goes extinct in isolation if $\lim_{c\rightarrow \infty} X_i\left(t_c^f\right) =0$ when $X_{j0}=0$ for $1\leq j\leq n,$ $j\neq i.$
    \item  The total population goes extinct (total extinction) if strain $S_i$ goes extinct for $1\leq i\leq n.$
    \item  Strain $S_i$ survives if $\lim_{c\rightarrow \infty} X_i\left(t_c^f\right) = X_{eq,i}>0,$ where $X_{eq,i}$ is the equilibrium concentration for $S_i.$
\end{enumerate}
\end{definition}

The extinction of strain $S_i$ in isolation is a particular case of extinction for which $S_i$ undergoes cyclic treatment in the absence of $\{S_j\mid 1\leq j\leq n,\;j\neq i\},$ i.e., independently of the competition with other strains. This concept will serve as a basis for studying the selection dynamics of the strains in competence. Note that Eq.~\eqref{eq:t_sat} has a closed-form solution for the saturation time when strain $S_i$ evolves in isolation, so that the time spent by the strain in the exponential growth phase at the $c$-cycle is:
\begin{align}\label{eq:tsat1_isol}
t_{i,c} = \min\{t_g,t_{i,c}^{sat}\},\quad t_{i,c}^{sat} = \frac{1}{\mu_i}\log\left(\frac{K}{X_{i,c}^0}\right),\quad c\geq 1,\quad 1\leq i\leq n.
\end{align}

Therefore, if strain $S_i$ is subjected to the cyclic protocol in isolation, the expression in Eq.~\eqref{eq:Xf} for the cell concentration at the end of the $c$-cycle takes the form:
\begin{align}\label{eq:Xf_isol}
    X_i(t_c^f)=D^c X_{i0}\exp\left(\mu_i\sum_{q=1}^c t_{i,q} I_e(X_{i,q}^0) - c k_i(C)t_k\right),\quad c\geq 1,\quad 1\leq i\leq n,
\end{align}
and, similarly, Eq.~\eqref{eq:Xf_fun_X0} relating the cell concentration at the initial and final time of each cycle transforms to:
\begin{align}\label{eq:Xf_fun_X0_isol}
    X_i(t_c^f)=&\;D X_{i,c}^0\exp\left(\mu_i t_{i,c} I_e(X_{i,c}^0) - k_i(C)t_k\right),\quad c\geq 1,\quad 1\leq i\leq n.
\end{align}  


First, we will proceed to analyse the conditions for total population extinction. 
If none of the strains survive cyclic treatment, the antimicrobial's selective pressure is enough to eliminate the population and prevent bacterial selection. Then, the selective part relying on cell survival under antimicrobial effect mainly determines the treatment outcome, and it makes no sense to study the selection component driven by bacterial competition (which is determined by the selection conditions, as we will see).

To obtain the conditions for total population extinction, we need first the following preliminary results, derived directly from Def.~\ref{def:ext} and Eqs.~\eqref{eq:ODEs_simple_model}--\eqref{fun:I_k}.

\begin{lemma}\label{lemma:tsat_isol}
   The saturation time for strain $S_i$ $(1\leq i\leq n)$ under cyclic treatment in isolation verifies:
    \begin{align*}
        t_c^{sat} \leq t_{i,c}^{sat},\quad c\geq 1.
    \end{align*}
\end{lemma}
\begin{proof}
   It should simply be noted that:
    \begin{align*}
        X_{i,c}^0\exp\left(\mu_i I_e(X_{i,c}^0)t_c^{sat}\right) + \sum_{l=1,l\neq i}^n X_{l,c}^0\exp\left(\mu_l I_e(X_{l,c}^0)t_c^{sat}\right) \underbrace{=}_{\text{Eq.~\eqref{eq:t_sat}}} K \underbrace{=}_{\text{Eq.~\eqref{eq:tsat1_isol}}} X_{i,c}^0\exp\left(\mu_i I_e(X_{i,c}^0) t_{i,c}^{sat}\right),\quad c\geq 1,
    \end{align*}
    and, since $X_{l,c}^0\geq 0$ for any strain $S_l$ $(1\leq l\leq n),$ we have:
    \begin{align*}
        X_{i,c}^0\exp\left(\mu_i I_e(X_{i,c}^0)t_c^{sat}\right) \leq  X_{i,c}^0\exp\left(\mu_i I_e(X_{i,c}^0) t_{i,c}^{sat}\right)\Rightarrow t_c^{sat} \leq t_{i,c}^{sat},\quad c\geq 1.
    \end{align*}
\end{proof}

\begin{lemma}\label{lemma:EC_isol}
   If strain $S_i$ $(1\leq i\leq n)$ goes extinct  under cyclic treatment in isolation, then it goes extinct.
\end{lemma}
\begin{proof}
    If strain $S_i$ goes extinct in isolation, using Def.~\ref{def:ext} and Eq.~\eqref{eq:Xf_isol}, we have:
    \begin{align*}
       \lim_{c\rightarrow \infty} D^c X_{i0}\exp\left(\mu_i\sum_{q=1}^c t_{i,q} I_e(X_{i,q}^0) - c k_i(C) t_k\right) = 0,
    \end{align*}
    and the proof is a direct consequence of Lem.~\ref{lemma:tsat_isol} and the Sandwich 
    Rule, because:
    \begin{align*}
       0\leq \lim_{c\rightarrow \infty} X_i(t_c^f) \underbrace{=}_{\text{Eq.~\eqref{eq:Xf}}} &\lim_{c\rightarrow \infty} D^c X_{i0}\exp\left(\mu_i\sum_{q=1}^c t_q I_e(X_{i,q}^0) - c k_i(C) t_k\right)  \underbrace{\leq}_{\text{Lem.~\ref{lemma:tsat_isol}}} \\
       &\lim_{c\rightarrow \infty} D^c X_{i0}\exp\left(\mu_i\sum_{q=1}^c t_{i,q} I_e(X_{i,q}^0) - c k_i(C) t_k\right)  = 0\Rightarrow \lim_{c\rightarrow \infty} X_i(t_c^f)=0.
    \end{align*}
\end{proof}

\begin{lemma}\label{lemma:ext_Xg}
    Strain $S_i$  $(1\leq i\leq n)$ goes extinct  under cyclic treatment if and only if $X_{i0}<X_e$ or there exist a cycle $c_{e,i}\geq 1$ such that:
    $$X_i(t_{c_{e,i}}^f) < X_e\leq X_i(t_{c_{e,i}}^0),$$ 
    where  $c_{e,i}$ is called the extinction cycle of strain $S_i.$
\end{lemma}
\begin{proof}
\textcolor{white}{.}\\[0.2cm]
$``\Rightarrow"$ It follows directly from Def.~\ref{def:ext}.\\[0.1cm]
$``\Leftarrow"$
    If strain $S_i$ $(1\leq i\leq n)$ is such that $X_{i0} < X_e,$ we can use Eq.~\eqref{eq:Xf} with $I_e(X_{i,c}^0)=0$ for $c\geq 1,$ to get:
    \begin{align*}
       \lim_{c\rightarrow\infty} X_i\left(t_c^f\right) \underbrace{=}_{\text{Eq.~\eqref{eq:Xf}}} \lim_{c\rightarrow\infty}
        D^c X_{i0}\exp\left(\mu_i\sum_{q=1}^c t_qI_e\left(X_{i,q}^0\right) - c k_i(C)t_k\right)=
        \lim_{c\rightarrow\infty} D^c X_{i0}\exp(-ck_i(C)t_k)=0.
    \end{align*}

    Then, if $X_{i0} < X_e,$ we take $c_{e,i}=0.$ 
    Conversely, if $X_i(t_{c_{e,i}}^f)<X_e\leq X_i(t_{c_{e,i}}^0)$ for some cycle $c_{e,i}\geq 1,$  we can use Eq.~\eqref{eq:Xf} by taking into account that:
    \begin{align*}
    I_e(X_{i,c}^0)=\begin{cases} 
    1,&c\leq c_{e,i},\\
    0,&c> c_{e,i},\end{cases}\quad c\geq 1,
    \end{align*}
    to obtain:
    \begin{align*}
       \lim_{c\rightarrow\infty} X_i\left(t_c^f\right) \underbrace{=}_{\text{Eq.~\eqref{eq:Xf}}}& \lim_{c\rightarrow\infty} D^cX_{i0}\exp\left(\mu_i\sum_{q=1}^{c}t_qI_e\left(X_{i,q}^0\right) - ck_i(C)t_k\right) =\\
       &\lim_{c\rightarrow\infty} D^cX_{i0}\exp\left(\mu_i\sum_{q=1}^{c_{e,i}}t_q - ck_i(C)t_k\right) = 
       X_{i0}\exp\left(\mu_i\sum_{q=1}^{c_{e,i}}t_q\right)\lim_{c\rightarrow\infty} D^c\exp\left( - ck_i(C)t_k\right)=0.
    \end{align*}
\end{proof}


We are now ready to derive necessary and sufficient conditions for total population extinction in terms of the model parameters. Then, the following result allows us to determine the success of cyclic antimicrobial protocols in bacterial populations with different competing strains. More interestingly, Prop.~\ref{prop:EC} can be used to determine the setup parameters for cyclic treatment that prevent bacterial selection. 

\begin{prop}[\textbf{Extinction conditions}]\label{prop:EC}
    Consider a bacterial population formed by strains $\{S_i\mid 1\leq i\leq n\}$ subjected to cyclic antimicrobial treatment, as described by Eqs.~\eqref{eq:ODEs_simple_model}--\eqref{fun:I_k}. Then:
    \begin{enumerate}[label=\arabic*)]
    \item Strain $S_i$ goes extinct in isolation if and only if $X_{i0}<X_e,$ or:
    \begin{align}\label{eq:EC_isol}
    t_{eq,i} > \min\left\{t_g, \frac{1}{\mu_i}\log\left(\frac{K}{X_e}\right)\right\},
    \end{align}
    where $t_{eq,i}$ is the equilibrium time for strain $S_i,$ defined in Lem.~\ref{lemma:lim_tp}.
    \item The total population goes extinct if and only if: 
    \begin{align*}
    J= \left\{1\leq j\leq n\mid X_e\leq X_{j0},\; t_{eq,j} \leq  \min\left\{t_g, t_1 - \frac{1}{\mu_j}\log\left(\frac{X_e}{X_{j0}}\right) \right\}\right\} = \varnothing.
    \end{align*}
    \end{enumerate}
\end{prop}

\begin{proof}
\textcolor{white}{b}\\
\vspace{-0.3cm}
    \begin{enumerate}[label=\arabic*)]
        \item 
        $``\Leftarrow"$ Let us assume that strain $S_i$ $(1\leq i\leq n)$ is subjected to cyclic treatment in isolation, so that $X_{j0}=0$ for $1\leq j\leq n,$ $j\neq i.$ If $X_{i0}<X_e,$ strain $S_i$ goes extinct by Lem.~\ref{lemma:ext_Xg}. Let us then consider $X_{i0}\geq X_e,$ and:
        \begin{align*}
            \frac{1}{\mu_i}\log\left(\frac{K}{X_e}\right) < t_{eq,i}=\frac{k_i(C)t_k - \log(D)}{\mu_i}.
        \end{align*}
        Then, we have:
        \begin{align*}
          \log\left(\frac{K}{X_e}\right) + \log(D) < k_i(C)t_k \Rightarrow 
          X_{i}(t_1^f) \underbrace{\leq}_{\text{Eq.~\eqref{eq:Xf_isol}}} DK\exp\left(-k_i(C)t_k\right) < X_e,
        \end{align*}
        and strain $S_i$ goes extinct by Lem.~\ref{lemma:ext_Xg}, with extinction cycle $c_{e,i}=1.$ On the other hand, if:
         \begin{align*}
        t_g < t_{eq,i},
        \end{align*}   
        and considering that the duration of the exponential growth phase for strain $S_i$ in isolation verifies:
        \begin{align*}
          \mu_i t_{i,c} I_e(X_{i,c}^0) + \log(D) - k_i(C)t_k\underbrace{\leq}_{\text{Eq.~\eqref{eq:tsat1_isol}}}  \mu_i t_g + \log(D) - k_i(C)t_k,\quad c\geq 1,
        \end{align*}
        then, by the sandwich rule, we obtain:
        \begin{align*}
        t_g < t_{eq,i}\Rightarrow&\; \mu_i t_g + \log(D) < k_i(C)t_k\Rightarrow \\
        &\;0\leq \lim_{c\rightarrow \infty} X_i(t_c^f) \underbrace{=}_{\text{Eq.~\eqref{eq:Xf_isol}}}\lim_{c\rightarrow \infty} X_{i0} \exp\left(\mu_i \sum_{q=1}^c t_{i,q}I_e(X_{i,q}^0) + c\log(D) - c k_i(C)t_k\right)\Rightarrow \\
        & \;0\leq \lim_{c\rightarrow \infty} X_i(t_c^f) \leq  \lim_{c\rightarrow \infty} X_{i0} \exp\left(c\left(\mu_i t_g + \log(D) - k_i(C)t_k\right)\right)=0\Rightarrow  \lim_{c\rightarrow \infty} X_i(t_c^f) =0.
        \end{align*}

        $``\Rightarrow"$ Let us assume that strain $S_i$ $(1\leq i\leq n)$ goes extinct in isolation,  i.e., when $X_{j0}=0$ for $1\leq j\leq n,$ $j\neq i,$ and consider $X_{i0}\geq X_e.$ By Lem.~\ref{lemma:ext_Xg}, there exist a cycle $c_{e,i}\geq 1$ such that:
        \begin{align*}
            X_i(t_{c_{e,i}}^f) < X_e \leq X_i(t_{c_{e,i}}^0) = X_{i,c_{e,i}}^0.
        \end{align*}

        Then, if the extinction cycle of strain $S_i$ in isolation is $c_{e,i}=1,$ we obtain:
        \begin{align*}
             X_i(t_1^f) < X_e,\;\; t_{i,1}^{sat}\leq t_g \underbrace{\Rightarrow}_{\text{Eq.~\eqref{eq:Xf_isol}}} &\; D K \exp(-k_i(C)t_k) < X_e\Rightarrow t_{eq,i}>\frac{1}{\mu_i}\log\left(\frac{K}{X_e}\right),
             \\
            X_i(t_1^f) < X_e\leq X_{i0},\;\; t_{i,1}^{sat} > t_g \underbrace{\Rightarrow}_{\text{Eq.~\eqref{eq:Xf_fun_X0_isol}}} &\; DX_{i0}\exp(\mu_it_g - k_i(C)t_k) <  X_{i0}\Rightarrow  t_{eq,i}>t_g,
        \end{align*}
        so that a necessary condition for $c_{e,i}=1$ is:
        \begin{align*}
            t_{eq,i} > \min\left\{t_g, \frac{1}{\mu_i}\log\left(\frac{K}{X_e}\right)\right\}.
        \end{align*}        
        
        On the contrary, if $c_{e,i}>1,$ note that strain $S_i$ in isolation cannot reach the saturation concentration at the extinction cycle. Otherwise, by Lem.~\ref{lemma:ext_Xg} it would be:
        \begin{align*}
          t_{i,c_{e,i}}^{sat}\leq t_g,\; c_{e,i}>1\Rightarrow  X_i(t_{c_{e,i}}^f)\underbrace{=}_{\text{Eq.~\eqref{eq:Xf_isol}}}&DK\exp\left(-k_i(C)t_k\right) < X_e \leq X(t_1^f) \Rightarrow\\
          &DK\exp\left(-k_i(C)t_k\right) < X_e \leq  X(t_1^f) \leq DK\exp\left(-k_i(C)t_k\right),
        \end{align*}
        which is impossible. Therefore, we have $t_{i,c_{e,i}}^{sat}>t_g$ when $c_{e,i} > 1,$ and a necessary condition for strain $S_i$ going extinct in isolation is:
        \begin{align*}
            X_i(t_{c_{e,i}}^f) < X_{i,c_{e,i}}^0,\;c_{e,i}>1  \Rightarrow X_i(t_{c_{e,i}}^f)\underbrace{=}_{\text{Eq.~\eqref{eq:Xf_fun_X0_isol}}}D X_{i,c_{e,i}}^0 \exp\left(\mu_i t_g - k_i(C)t_k\right) < X_{i,c_{e,i}}^0 \Rightarrow  t_{eq,i}>t_g.
        \end{align*}

       \item $``\Leftarrow"$ Let us assume that $J=\varnothing.$ Then, for any strain $S_i$ $(1\leq i\leq n),$ we have $X_{i0}<X_e,$ or:
        \begin{align*}
           t_{eq,i} > \min\left\{t_g, t_1 - \frac{1}{\mu_i}\log\left(\frac{X_e}{X_{i0}}\right) \right\},\quad 1\leq i\leq n.
        \end{align*}
        
       First, we just demonstrated in Prop.~\ref{prop:EC}.1 that strain $S_i$ goes extinct in isolation if $X_{i0} < X_e,$ or it holds that:
        \begin{align*}
            t_{eq,i} > t_g,
        \end{align*}   
        in which case strain $S_i$ goes extinct by Lem.~\ref{lemma:EC_isol}. 
        Conversely, if the above condition does not hold, strain $S_i$ should verify:
        \begin{align*}
           t_1 - \frac{1}{\mu_i}\log\left(\frac{X_e}{X_{i0}}\right) < t_{eq,i} \leq t_g,
        \end{align*}
        so that $S_i$ goes extinct by Lem.~\ref{lemma:ext_Xg}, with extinction cycle $c_{e,i}=1.$ Indeed, we have:
        \begin{align*}
            t_1 - \frac{1}{\mu_i}\log\left(\frac{X_e}{X_{i0}}\right) < t_{eq,i}\Rightarrow 
           X_i(t_1^f) \underbrace{=}_{\text{Eq.~\eqref{eq:Xf}}} D X_{i0}\exp\left(\mu_i t_1 - k_i(C) t_k\right) < X_e.
        \end{align*}

        In conclusion, the population becomes extinct when all strains go extinct in isolation, or when the strains surviving in isolation fall below the extinction limit after $c=1$ due to competition with the other strains. Note that the last scenario can only occur when $t_1<t_{i,1}$ for the strains $S_i$ surviving in isolation, i.e., only when $t_1=t_1^{sat} <t_g,$ and $X_{l0}>0$ for some $1\leq l\leq n,$ $l\neq i.$
        
        $``\Rightarrow"$ Let us assume the total population goes extinct under cyclic treatment, i.e., all strains $\{S_i\mid 1\leq i\leq n\}$ go extinct, and consider $J\neq \varnothing.$ Then, there exist some strain $S_i$ with $i\in J,$ and:
        \begin{align*}
            t_{eq,i} = \min\{t_{eq,j}\mid j\in J\}.
        \end{align*}
        
         Since strain $S_i$ goes extinct, by Lem.~\ref{lemma:ext_Xg} there exist a cycle $c_{e,i}\geq 1$ such that:
        \begin{align*}
            X_i(t_{c_{e,i}}^f) < X_e \leq  X_i(t_{c_{e,i}}^0)=X_{i,c_{e,i}}^0,
        \end{align*}
        and the extinction cycle for $S_i$ verifies $c_{e,i}>1,$ because:
        \begin{align*}
            i\in J\Rightarrow X_e\leq X_{i0},\quad t_{eq,i} < t_1 - \frac{1}{\mu_i}\log\left(\frac{X_e}{X_{i0}}\right) \Rightarrow X_i(t_1^f)\underbrace{=}_{\text{Eq.~\eqref{eq:Xf}}}DX_{i0}\exp\left(\mu_i t_1 - k_i(C)t_k\right)\geq X_e.
        \end{align*}

        Now, at the extinction cycle $c_{e,i}>1,$ it holds that:
        \begin{align*}
             X_i(t_{c_{e,i}}^f) <   X_{i,c_{e,i}}^0 \underbrace{\Rightarrow}_{\text{Eq.~\eqref{eq:Xf_fun_X0}}} t_{c_{e,i}} < t_{eq,i}=\min\{t_{eq,j}\mid j\in J\} \leq t_{eq,j},\quad j\in J,
        \end{align*}
        and the population should reach the carrying capacity at the extinction cycle, because:
        \begin{align*}
            i\in J,\quad X_i(t_{c_{e,i}}^f) <   X_{i,c_{e,i}}^0,\quad c_{e,i}>1\underbrace{\Rightarrow}_{\text{Eq.~\eqref{eq:Xf_fun_X0}}} t_{c_{e,i}} < t_{eq,i} \leq t_g.
        \end{align*}

        Then, the total cell concentration at the end of the growth period $c_{e,i}-1\geq 1$ verifies:
        \begin{align*}
            \sum_{l=1}^n X_l(t_{c_{e,i}-1}^0)\exp\left(\mu_l t_{c_{e,i}-1}I_e(X_{l,c_{e,i}-1}^0)\right) \leq  K = \sum_{l=1}^n X_l(t_{c_{e,i}}^0)\exp\left(\mu_l t_{c_{e,i}}I_e(X_{l,c_{e,i}}^0)\right),
        \end{align*}
        so that we obtain:
        \begin{align*}
            \sum_{l=1}^n X_l(t_{c_{e,i}}^0)\frac{\exp(k_lt_k)}{D}\underbrace{=}_{\text{Eq.~\eqref{eq:Xf_fun_X0}}}&\sum_{l=1}^n X_l(t_{c_{e,i}-1}^0)\exp\left(\mu_l t_{c_{e,i}-1}I_e(X_{l,c_{e,i}-1}^0)\right) \leq \\
            &\sum_{l=1}^n X_l(t_{c_{e,i}}^0)\exp\left(\mu_l t_{c_{e,i}}I_e(X_{l,c_{e,i}}^0)\right)\Rightarrow\\
            &0\leq \sum_{l=1}^n X_l(t_{c_{e,i}}^0)\left(\exp\left(\mu_l t_{c_{e,i}} I_e(X_{l,{c_{e,i}}}^0)\right) - \frac{\exp(k_l t_k)}{D}\right).
         \end{align*}
         
         Consequently, there exist some strain $S_l$ $(1\leq l\leq n)$ such that:
        \begin{align*}
        \frac{\exp(k_l t_k)}{D}\leq \exp\left(\mu_l t_{c_{e,i}}I_e(X_{l,c_{e,i}}^0)\right)\leq \exp\left(\mu_l t_{c_{e,i}}\right) \Rightarrow I_e(X_{l,c_{e,i}}^0)=1,\quad t_{eq,l}\leq t_{c_{e,i}}<t_g.
        \end{align*}
        and $l\in J,$ since strain $S_l$ is above the extinction limit at the start of cycle $c_{e,i}>1,$ and $t_{eq,l}<t_g.$ However, this is a contradiction, because:
        \begin{align*}
            t_{eq,l}\leq t_{c_{e,i}}< t_{eq,j},\quad j\in J.
        \end{align*}

        Therefore, strain $S_i$ cannot go extinct when $i\in J,$ and it should be $J=\varnothing.$
        \end{enumerate}  
\end{proof}

\begin{remark}
    The condition for total population extinction in Prop.~\ref{prop:EC}.2 involves the saturation time of the population at cycle $c=1,$ which has no closed-form expression as a function of the model parameters. Then, the condition for total population extinction should be approximated using a numerical solution of Eq.~\ref{eq:t_sat}, or the approximation provided in \nameref{sec:MathModel:subsec:Approx_tsat}.
\end{remark}

The condition in Prop.~\ref{prop:EC}.1 means that extinction under cyclic treatment in isolation depends on the equilibrium time for the strains. That is, a strains goes extinct independently of competition with the others if an equilibrium cannot be reached between new bacteria formed during growth and cell removal after killing and dilution. Then, a strain eventually extinguishes in isolation if, in each cycle, more bacteria are killed and diluted than formed during exponential growth. Additionally, the strains can be eliminated in isolation with just one cycle if antimicrobial killing and dilution are strong enough to reduce the cell concentration from the carrying capacity to values below the extinction limit.

On the other hand, the condition for total population extinction derived in Prop.~\ref{prop:EC}.2  states that the strains go extinct under the cyclic protocol in competition if they cannot survive in isolation. Nevertheless, the bacterial population can also become extinct if competition during the first growth period impedes the propagation of the strains that would otherwise survive in isolation. That is an interesting result, particularly in the case where the population is formed by an ancestor strain, together with resistant or tolerant mutants. Indeed, if the mutant strain is initially present at very low frequencies in the population (the typical situation), so that it cannot reproduce enough during the first growth period due to the strong competition with the (much more frequent) ancestor, then the mutant can be removed from the population in just one cycle, even when it would outcompete the ancestor in other case.

If the condition in Prop.~\ref{prop:EC}.2 does not hold and the bacterial population survives the cyclic protocol, the next step is to use Eqs.~\eqref{eq:ODEs_simple_model}--\eqref{fun:I_k} to determine which strains have the selective advantage and  outcompete the others. 

To derive the selection conditions, we need the following preliminary results.
\begin{lemma}\label{lemma:lim_tp}
    If strain $S_i$  $(1\leq i\leq n)$ survives cyclic treatment, then:
    \begin{align*} 
    \lim_{c\rightarrow \infty} t_c = t_{eq,i}= \frac{k_i(C)t_k - \log(D)}{\mu_i},
    \end{align*}
    where $t_{eq,i}$ is called the equilibrium time for strain $S_i.$
\end{lemma}
\begin{proof}

If strain $S_i$  $(1\leq i\leq n)$ survives, using Lem.~\ref{lemma:ext_Xg} we have:
\begin{align*}
    X_i(t_c^f) \geq X_e,\quad c\geq 1,
\end{align*}
and, by Def.~\ref{def:ext}, it holds that:
\begin{align*}
\lim_{c\rightarrow \infty} X_i(t_{c+1}^f)=\lim_{c\rightarrow \infty}X_i\left(t_c^f\right)=X_{eq,i}.
\end{align*}

Then, we can use Eq.~\eqref{eq:Xf_fun_X0} with $I_e(X_{i,c}^0)=1$ for $c\geq 1,$ to get:
\begin{align*}
    X_{eq,i}=&\lim_{c\rightarrow \infty}X_i(t_{c+1}^f)\underbrace{=}_{\text{Eq.~\eqref{eq:Xf_fun_X0}}} \lim_{c\rightarrow \infty}DX_{i,c+1}^0\exp\left(\mu_i t_{c+1} - k_i(C)t_k\right)\underbrace{=}_{\text{Rmk.~\ref{rem:limXf}}}
    \lim_{c\rightarrow \infty}DX_i(t_c^f)\exp\left(\mu_i t_{c+1} - k_i(C)t_k\right)=\\
    &X_{eq,i}\lim_{c\rightarrow \infty}D\exp\left(\mu_i t_{c+1} - k_i(C)t_k\right).
\end{align*}

Finally, dividing both sides of the above equation by $X_{eq,i}>0,$ and operating with the limits, we obtain:
\begin{align*}
\lim_{c\rightarrow \infty}D\exp\left(\mu_i t_{c+1} - k_i(C)t_k\right)=1 \Leftrightarrow \lim_{c\rightarrow \infty} \mu_i t_{c+1} = k_i(C)t_k - \log(D),
\end{align*}
so that, taking into account that the growth rate of strain $S_i$ is distinct from zero $(\mu_i>0),$  we arrive at:
\begin{align*}
   \lim_{c\rightarrow\infty} t_{c+1} =  \lim_{c\rightarrow\infty} t_c = \frac{k_i(C)t_k - \log(D)}{\mu_i}.
\end{align*}
\end{proof}

\begin{corollary}\label{cor:lim_tp_WM}
    If strains $S_i$ and $S_j$ survive cyclic treatment $(1\leq i,j\leq n,$ $j\neq i),$ then:
    \begin{align*}
        t_{eq,i}=t_{eq,j}.
    \end{align*}
\end{corollary}
\begin{proof}
    The demonstration follows directly from Lem.~\ref{lemma:lim_tp}.
\end{proof}

We are now ready to derive necessary and sufficient conditions on the model parameters for determining which strains survive cyclic treatment.

\begin{prop}[\textbf{Selection conditions}]\label{prop:SC}
    Consider a bacterial population formed by strains $\{S_i\mid 1\leq i\leq n\}$ subjected to cyclic antimicrobial treatment, as described by Eqs.~\eqref{eq:ODEs_simple_model}--\eqref{fun:I_k}, and define:
    \begin{align*}
    J= \left\{1\leq j\leq n\mid X_e\leq X_{j0},\; t_{eq,j} \leq  \min\left\{t_g, t_1 - \frac{1}{\mu_j}\log\left(\frac{X_e}{X_{j0}}\right) \right\}\right\},
    \end{align*}
    as in Prop.~\ref{prop:EC}.2. Then, strain $S_i$ survives if and only if $i\in J,$ and:
     \begin{align*}
         t_{eq,i} = \min\left\{t_{eq,j}\mid j\in J\right\}.
    \end{align*}    
\end{prop}
\begin{proof}
   \textcolor{white}{b}\\[0.2cm]
        $``\Leftarrow"$ Let us assume that strain $S_i$ $(1\leq i\leq n)$ verifies $i\in J$ and $t_{eq,i}\leq t_{eq,j}$ for $j\in J.$ Then, if strain $S_i$ goes extinct under cyclic treatment, by Lem.~\ref{lemma:ext_Xg} there exist a cycle $c_{e,i}$ such that:
        \begin{align*}
            X_i(t_{c_{e,i}}^f) < X_e \leq X_i(t_{c_{e,i}}^0)\underbrace{\Rightarrow}_{Eq.~\eqref{eq:Xf_fun_X0}} &\mu_i t_{c_{e,i}} - k_it_k + \log(D) < 0 \Rightarrow \\
            &t_{c_{e,i}} < t_{eq,i}= \frac{k_it_k - \log(D)}{\mu_i} \leq t_{eq,j},\quad j\in J,
        \end{align*}
        and $c_{e,i}>1,$ as $i\in J.$ Additionally, $i\in J$ implies $t_{eq,i}\leq t_g$ so that, since $S_i$ goes extinct, we obtain:
        \begin{align*}
         t_{c_{e,i}} < t_{eq,i}\leq t_g,
        \end{align*}
       and the population reaches the saturation concentration at the extinction cycle of strain $S_i.$ That is:
        \begin{align*}
          t_{c_{e,i}}=t_{c_{e,i}}^{sat} < t_g \Rightarrow  K = \sum_{l=1}^n X_l(t_{c_{e,i}}^0)\exp\left(\mu_l t_{c_{e,i}}I_e(X_{l,c_{e,i}}^0)\right).
        \end{align*}
        
        Then, the total cell concentration at the end of the growth period $c_{e,i}-1\geq 1$ verifies:
        \begin{align*}
            \sum_{l=1}^n X_l(t_{c_{e,i}-1}^0)\exp\left(\mu_l t_{c_{e,i}-1}I_e(X_{l,c_{e,i}-1}^0)\right) \leq  K = \sum_{l=1}^n X_l(t_{c_{e,i}}^0)\exp\left(\mu_l t_{c_{e,i}}I_e(X_{l,c_{e,i}}^0)\right),
        \end{align*}
        so that:
        \begin{align*}
            \sum_{l=1}^n X_l(t_{c_{e,i}}^0)\frac{\exp(k_lt_k)}{D}\underbrace{=}_{\text{Eq.~\eqref{eq:Xf_fun_X0}}}&\sum_{l=1}^n X_l(t_{c_{e,i}-1}^0)\exp\left(\mu_l t_{c_{e,i}-1}I_e(X_{l,c_{e,i}-1}^0)\right) \leq \\
            &\sum_{l=1}^n X_l(t_{c_{e,i}}^0)\exp\left(\mu_l t_{c_{e,i}}I_e(X_{l,c_{e,i}}^0)\right)\Rightarrow \\
            &0\leq \sum_{l=1}^n X_l(t_{c_{e,i}}^0)\left(\exp\left(\mu_l t_{c_{e,i}} I_e(X_{l,{c_{e,i}}}^0)\right) - \frac{\exp(k_l t_k)}{D}\right).
         \end{align*}
         
         However, if the above condition holds, there exists a strain $S_m$ $(1\leq m\leq n)$ such that:
        \begin{align*}
        \frac{\exp(k_m t_k)}{D}\leq \exp\left(\mu_m t_{c_{e,i}}I_e(X_{m,c_{e,i}}^0)\right) \Rightarrow I_e(X_{m,c_{e,i}}^0)=1,\quad t_{eq,m}\leq t_{c_{e,i}}<t_g,
        \end{align*}
         and it should be $m\in J,$ since $S_m$ is above the extinction limit at the start of cycle $c_{e,i}>1,$ and $t_{eq,m}<t_g.$ We thus arrive at a contradiction, because:
        \begin{align*}
            t_{eq,m}\leq t_{c_{e,i}}<t_{eq,i}\leq t_{eq,j},\quad m,j\in J.
        \end{align*}
        
        $``\Rightarrow"$ Let us assume that strain $S_i$ $(1\leq i\leq n)$ survives under cyclic treatment. Then, the total population cannot become extinct and, by Prop.~\ref{prop:EC}, it should be $i\in J\neq \varnothing.$ Now, if there exist a strain $S_l$ with $l\in J,$ such that:
        \begin{align*}
            t_{eq,l}=\min\{t_{eq,j}\mid j\in J\} < t_{eq,i},
        \end{align*}
         we just demonstrated that $S_l$ survives under cyclic treatment. However, taking into account Cor.~\ref{cor:lim_tp_WM}, we would have:
        \begin{align*}
        t_{eq,i} = t_{eq,l}< t_{eq,i},
        \end{align*}
        which is impossible. Therefore, if strain $S_i$ survives, it should be  $t_{eq,i}\leq t_{eq,j}$ for any $j\in J.$
\end{proof}

When an ancestor bacterial strain competes under cyclic treatment with a resistant or tolerant mutant, the selection condition in Prop.~\ref{prop:SC} can be represented as a trade-off between survival advantage $(\text{SA})$ and fitness cost $(\text{FC})$ conferred by the mutation. Indeed, define the fitness cost of the mutant $S_j$ as the reduction in growth rate relative to the ancestor $S_i$ $(1\leq i,j\leq n,$  $j\neq i),$ that is:
\begin{align*}
    \text{FC}_j = \frac{\mu_i - \mu_j}{\mu_i} = 1 - \frac{\mu_j}{\mu_i},
\end{align*}
and the survival advantage as the increase in cell survival after killing and dilution of the mutant relative to the ancestor, so that:
\begin{align*}
    \text{SA}_j = \frac{\log\left(X_i(t_c^f)/X_i(t_c^f-t_k)\right) - \log\left(X_j(t_c^f)/X_j(t_c^f-t_k)\right)}{\log(X_i(t_c^f)/X_i(t_c^f-t_k))} = 1 - \frac{k_jt_k - \log(D)}{k_i t_k - \log(D)}.
\end{align*}

Then, the condition for selective equilibrium between ancestor $S_i$ and mutant $S_j$ given in Prop.~\ref{prop:SC} takes the form:
\begin{align*}
   t_{eq,i} = t_{eq,j}\Leftrightarrow \text{FC}_j = \text{SA}_j.
\end{align*}

The above result indicates that mutant strains $S_j$ that cannot compensate for the fitness cost by increasing cell survival under antimicrobial killing (i.e., if $\text{FC}_j> \text{SA}_j)$ will be outcompeted by the ancestor. Conversely, the mutant will eventually outcompete the ancestor when the survival advantage conferred  by the mutation exceeds the fitness cost (i.e.,  $\text{FC}_j< \text{SA}_j).$ That is true as long as the mutant $S_j$  is not eliminated by competition with the remaining strains during the first cycle of treatment (see Prop.~\ref{prop:EC}.2).

Note that the selection condition in Prop.~\ref{prop:SC} only informs about which strain is eventually selected under the cyclic protocol, but does not provide any measure of the selective pressure exerted by one strain over another. To do that, we can use the following definition, which provides a quantitative measure of the selective advantage or disadvantage of one strain over the others.


\begin{definition}
Consider a bacterial population formed by strains $\{S_i\mid 1\leq i\leq n\}$ subjected to cyclic antimicrobial treatment. The selection coefficient of strain $S_i$ over $S_j$ at the $c$-cycle is~\cite{2011Chevin,2020Lin}:
\begin{align}\label{eq:sel_coeff}
    s_{i\mid j,c} = \log\left(\frac{X_i(t_c^f)}{X_i(t_c^0)}\right) -  \log\left(\frac{X_j(t_c^f)}{X_j(t_c^0)}\right),\quad c\geq 1,\quad 1\leq i,j\leq n, \quad i\neq j.
\end{align}
\end{definition}

\begin{remark} The selection coefficient verifies:
   \begin{align*}
   s_{i \mid j,c} = - \left( \log\left(\frac{X_j(t_c^f)}{X_j(t_c^0)}\right) -  \log\left(\frac{X_i(t_c^f)}{X_i(t_c^0)}\right)\right) =  - s_{j \mid i,c},\quad c\geq 1,\quad 1\leq i,j\leq n, \quad i\neq j.
\end{align*}
\end{remark}

The selection coefficient in Eq.~\eqref{eq:sel_coeff} thus measures, at any cycle of treatment, the relative increase or decrease in cell concentration of one strain compared to another. Then, if the cell concentration of strain $S_i$ increased more than for $S_j$ from the start to the end of the $c$-cycle, the selection coefficient $s_{i\mid j,c}$ is positive  and increases with the difference in survival between the strains. Additionally, if the cell concentration of both strains decreases during the $c$-cycle, the selection coefficient $s_{i\mid j,c}$ is equally positive when strain $S_j$ is being eliminated by cyclic treatment faster than  $S_i.$ In consequence,  the selection coefficient can compare the selective ability between any pair of competing strains, but it cannot detect whether both strains go extinct.

For the model of cyclic antimicrobial treatment in Eqs.~\eqref{eq:ODEs_simple_model}--\eqref{fun:I_k}, the selection coefficient of Eq.~\eqref{eq:sel_coeff} takes the form:
\begin{align}\label{eq:sel_coeff_sm}\nonumber
    s_{i\mid j,c} = & \log\left(\frac{X_i(t_c^f)}{X_i(t_c^0)}\right) -  \log\left(\frac{X_j(t_c^f)}{X_j(t_c^0)}\right) \underbrace{=}_{\text{Eq.~\eqref{eq:Xf_fun_X0}}}\\ \nonumber &\log\left(\frac{DX_{i,c}^0\exp(\mu_{i}t_c I_e(X_{i,c}^0) - k_{i}(C)t_k)}{X_{i,c}^0}\right) - \log\left(\frac{DX_{j,c}^0\exp(\mu_{j}t_c I_e(X_{j,c}^0) - k_{j}(C)t_k)}{X_{j,c}^0}\right)=\\[0.1cm]
    &\left(\mu_{i}I_e(X_{i,c}^0) - \mu_{j}I_e(X_{j,c}^0)\right)t_c - \left(k_{i}(C) - k_{j}(C)\right)t_k,\quad c\geq 1,\quad 1\leq i,j\leq n, \quad i\neq j.
\end{align}

\subsection{Extinction cycle}

Taking into account the extinction and selection conditions provided in Props.~\ref{prop:EC}--\ref{prop:SC}, we can now derive bounds on the extinction cycle for the strains that are eventually eliminated by cyclic treatment. These bounds can be used to analyse the effect of the setup parameters and bacterial traits on the extinction dynamics, or to determine the number of cycles required to eliminate the strains without running model simulations.

The following consequence of Prop.~\ref{prop:EC}.2 will be useful in the derivations that follow.
\begin{corollary}\label{cor:t_c=t_g}
    Consider a bacterial population formed by strains $\{S_i\mid 1\leq i\leq n\}$ subjected to cyclic antimicrobial treatment, as described by Eqs.~\eqref{eq:ODEs_simple_model}--\eqref{fun:I_k}. It holds that:
    \begin{enumerate}
        \item[1)]  If the total population goes extinct, the saturation cannot be reached from the first cycle; that is:
                    \begin{align*}
                        t_c = t_g,\quad c>1.
                    \end{align*}
        \item[2)]  If strain $S_i$ survives, the equilibrium time  of strain $S_i$ verifies:
                    \begin{align*}
                     t_{eq,i}\leq t_c,\quad c>1.
                    \end{align*}
    \end{enumerate}
\end{corollary}
\begin{proof}
\textcolor{white}{b}\\
\vspace{-0.3cm}
    \begin{enumerate}
        \item[1)] If the total population goes extinct under cyclic treatment, by Prop.~\ref{prop:EC}.2, we have that any strain $S_i$ $(1\leq i\leq n)$ goes extinct in isolation, or the extinction cycle verifies $c_{e,i}=1.$ Then, $t_{eq,i}>t_g$ for any strain $S_i$ such that $c_{e,i}>1.$ However, if the population reaches the saturation concentration at some cycle $c>1,$ we can follow a similar reasoning to that of Prop.~\ref{prop:EC}.2, and obtain:
        \begin{align*}
            \sum_{l=1}^n X_{l,c}^0&\exp\left(\mu_lt_cI_e(X_{l,c}^0)\right)\underbrace{=}_{\text{Eq.~\eqref{eq:t_sat}}} K \geq \sum_{l=1}^n X_{l,c-1}^0\exp\left(\mu_lt_{c-1}I_e(X_{l,c-1}^0\right) \underbrace{=}_{\text{Eq.~\eqref{eq:Xf_fun_X0}}}   \sum_{l=1}^n X_{l,c}^0\frac{\exp(k_l(C)t_k)}{D}\Rightarrow\\
            &\sum_{l=1}^n X_{l,c}^0\left(\exp\left(\mu_l t_c I_e(X_{l,c}^0)\right)-\frac{\exp(k_l(C)t_k)}{D}\right)\geq 0\Rightarrow  I_e(X_{m,c}^0)=1,\;\: t_{eq,m} \leq t_c,
         \end{align*}
        for some strain $S_m$ $(1\leq m\leq n).$ We thus arrive at a contradiction, since  $t_c < t_g < t_{eq,m}.$
        
        \item[2)] If strain $S_i$ survives cyclic treatment, using Prop.~\ref{prop:SC}, it holds that $i\in J$ and:
        \begin{align*}
            t_{eq,i} = \min\{t_{eq,j}\mid j\in J\},\quad J=\left\{1\leq j\leq n\mid X_e\leq X_{j0},\; t_{eq,j}\leq \min\left\{t_g,t_1 - \frac{1}{\mu_j}\log\left(\frac{X_e}{X_{j0}}\right)\right\}\right\}.
        \end{align*}
         
         Then, if there exists a cycle $c> 1$ such that $t_c<t_{eq,i},$ we have $t_c<t_{eq,i}\leq t_{eq,j}\leq t_g$
        for any $j\in J,$ and the population should reach the carrying capacity at the $c$-cycle. Also, strains $S_l$ $(1\leq l\leq n)$ such that $l\notin J$ should verify $c_{e,l}=1$ or $t_c<t_g<t_{eq,l},$ by Prop.~\ref{prop:EC}.2. We thus arrive at:
        \begin{align*}
          \sum_{l=1}^n X_{l,c}^0&\exp\left(\mu_lt_cI_e(X_{l,c}^0)\right)\underbrace{=}_{\text{Eq.~\eqref{eq:t_sat}}} K \geq \sum_{l=1}^n X_{l,c-1}^0\exp\left(\mu_lt_{c-1}I_e(X_{l,c-1}^0\right) \underbrace{=}_{\text{Eq.~\eqref{eq:Xf_fun_X0}}}   \sum_{l=1}^n X_{l,c}^0\frac{\exp(k_l(C)t_k)}{D}\Rightarrow\\
            &\sum_{l=1}^n X_{l,c}^0\left(\exp\left(\mu_l t_c I_e(X_{l,c}^0)\right)-\frac{\exp(k_l(C)t_k)}{D}\right)\geq 0\Rightarrow  I_e(X_{m,c}^0)=1,\;\: t_{eq,m} \leq t_c, 
        \end{align*}
         for some strain $S_m$ $(1\leq m\leq n),$ which is impossible.
    \end{enumerate}
\end{proof}

When the bacterial population goes extinct under cyclic treatment, we have demonstrated in Prop.~\ref{prop:EC}.2 that the strains surviving in isolation should fall below the extinction limit at the end of the first cycle due to competition. Additionally,   Cor.~\ref{cor:t_c=t_g} states that the population cannot reach the saturation concentration from the first cycle onwards, so that the strains surviving after the first cycle evolve as if they were in isolation. We thus have the following result.

\begin{corollary}\label{cor:ext_cyc_totalext}
    Consider a bacterial population formed by strains $\{S_i\mid 1\leq i\leq n\}$ subjected to cyclic antimicrobial treatment, as described by Eqs.~\eqref{eq:ODEs_simple_model}--\eqref{fun:I_k}. If the total population goes extinct, the extinction cycle of strain $S_i$ is:
    \begin{align}
        c_{e,i}=\begin{cases}
        0,& X_{i0} < X_e,\\
            1,&   X_{i0} \geq X_e,\;t_{eq,i}> t_1 - \log\left(X_e/X_{i0}\right)/\mu_i,\\
             \ceil*{\left(t_1 - \log\left(X_e/X_{i0}\right)/\mu_i - t_g\right)/\left(t_{eq,i} - t_g\right)},& \text{i.o.c.,}
        \end{cases}
    \end{align} 
    where $y=\ceil*{x}$ denotes the ceiling function, i.e., $y$ is the minimum integer greater than or equal to $x.$
\end{corollary}
\begin{proof}
    If the total population goes extinct, from Prop.~\ref{prop:EC}.2, we have $X_{i0}<X_e$ for any strain $S_i$ $(1\leq i\leq n),$ or:
    \begin{align*}
        t_{eq,i} > \min\left\{t_g, t_1 - \frac{1}{\mu_i}\log\left(\frac{X_e}{X_{i0}}\right)\right\}.
    \end{align*}

    If $X_{i0}<X_e,$ we can set the extinction cycle to $c_{e,i}=0,$ since the cell concentration of strain $S_i$ is already below the extinction limit before the treatment starts. On the other hand, Prop.~\ref{prop:EC}.2 states that $c_{e,i}=1$ when $X_{i0}\geq X_{e},$ and:
    \begin{align*}
    t_{eq,i} > t_1 - \frac{1}{\mu_i}\log\left(\frac{X_e}{X_{i0}}\right).
    \end{align*}

    Then, the unique non-trivial case is when $X_{i0}\geq X_e,$ and we have:
    \begin{align*}
        t_g < t_{eq,i} \leq t_1 - \frac{1}{\mu_i}\log\left(\frac{X_e}{X_{i0}}\right),
    \end{align*}    
    so that the extinction cycle of strain $S_i$ verifies $c_{e,i}>1.$ Then, since the duration of the exponential growth phase is $t_c=t_g$ for any cycle $c>1,$ by Cor.~\ref{cor:t_c=t_g}, and taking into account  Lem.~\ref{lemma:ext_Xg}, it holds that:
    \begin{align*}
      X_i(t_{c_{e,i}}^f) < X_e,\;\: I_e(X_{i,c}^0)=1,&\;\: t_c=t_g,\;\: 1<c\leq c_{e,i}\Rightarrow\\
      &X_i(t_c^f)\underbrace{=}_{\text{Eq.~\eqref{eq:Xf}}} D^{c_{e,i}}X_{i0}\exp\left(\mu_it_1 + \mu_i\left(c_{e,i}-1\right)t_g - c_{e,i}k_i(C)t_k\right) < X_e \Leftrightarrow\\
      &c_{e,i}\left(\log(D) + \mu_i t_g - k_i(C)t_k\right) < \log\left(\frac{X_e}{X_{i0}}\right) + \mu_i\left(t_g - t_1\right)\Leftrightarrow\\
      &c_{e,i}\left(\mu_it_g - \mu_i t_{eq,i}\right) < \log\left(\frac{X_e}{X_{i0}}\right) + \mu_i\left(t_g - t_1\right)\Leftrightarrow\\
      &c_{e,i}\left(t_g - t_{eq,i}\right) < t_g - \left(t_1 - \frac{1}{\mu_i}\log\left(\frac{X_e}{X_{i0}}\right)\right).
    \end{align*}

    Since strain $S_i$ verifies $t_g - t_{eq,i}<0,$ the above condition is equivalent to:
    \begin{align*}
      X_i(t_{c_{e,i}}^f) < X_e,\;\: I_e(X_{i,c}^0)=1,\;\:& t_c=t_g,\;\: 1<c\leq c_{e,i}\Leftrightarrow \\
      &I_e(X_{i,c}^0)=1,\;\: t_c=t_g,\;\: 1<c\leq c_{e,i},\;\: c_{e,i} >  \frac{t_1 - \log\left(X_e/X_{i0}\right)/\mu_i - t_g}{t_{eq,i} - t_g},
    \end{align*}
    which leads to the desired expression for the extinction cycle of strain $S_i:$
    \begin{align*}
         X_{i0}\geq X_e,\;\:t_g < t_{eq,i} \leq\; & t_1 - \frac{1}{\mu_i}\log\left(\frac{X_e}{X_{j0}}\right)\Rightarrow\\
         &c_{e,i} =\min\left\{c>1\mid c > \frac{t_1 - \log\left(X_e/X_{i0}\right)/\mu_i - t_g}{t_{eq,i} - t_g}\right\}= \ceil*{  \frac{t_1 - \log\left(X_e/X_{i0}\right)/\mu_i - t_g}{t_{eq,i} - t_g}}.
    \end{align*}  
    
\end{proof}

Let us now study the extinction cycle when the competition between strains extends beyond the first cycle of the antimicrobial treatment. In that case, the extinction dynamics of the strains depend on their ability to survive  cyclic treatment and outcompete the others. A closed-form expression for the extinction cycle, however, cannot be obtained when at least two strains are competing from cycle $c>1,$ since there is no closed-form expression for the saturation times in Eq.~\eqref{eq:t_sat}. Nevertheless, we can use the following result to bound the extinction cycle when at least one strain is eventually selected under cyclic treatment.
\begin{corollary}\label{cor:ext_cycle}
    Consider a bacterial population formed by strains $\{S_i\mid 1\leq i\leq n\}$ subjected to cyclic antimicrobial treatment, as described by Eqs.~\eqref{eq:ODEs_simple_model}--\eqref{fun:I_k}. If the total population survives, so that:
        \begin{align*}
        J= \left\{1\leq j\leq n\mid X_e\leq X_{j0},\; t_{eq,j} \leq  \min\left\{t_g, t_1 - \frac{1}{\mu_j}\log\left(\frac{X_e}{X_{j0}}\right) \right\}\right\}\neq \varnothing,
    \end{align*}
    then, it holds that:
    \begin{enumerate}
        \item[1)] If $t_g=\min\{t_{eq,j}\mid j\in J\},$ the extinction cycle for any strain $S_i$ going extinct is:
            \begin{align*}
        c_{e,i}=\begin{cases}
            0, & X_{i0} < X_e,\\
            1 ,&  X_{i0} \geq X_e,\;t_{eq,i} > t_1 - \log\left(X_e/X_{i0}\right)/\mu_i,\\
         \ceil*{\left(t_1 - \log\left(X_e/X_{i0}\right)/\mu_i - t_g\right)/\left(t_{eq,i} - t_g\right)},& \text{i.o.c.}
        \end{cases}
    \end{align*}    
        \item[2)] If $t_g>\min\{t_{eq,j}\mid j\in J\},$ the extinction cycle for any strain $S_i$ going extinct is bounded as:
        \begin{align*}
            L_{e,i}(l) < c_{e,i},
        \end{align*}
        for any strain $S_l$ $(1\leq l\leq n)$ such that $t_{eq,l} = \min\{t_{eq,j}\mid j\in J\},$ where:
        \begin{align*}
        L_{e,i}(l) &= \begin{cases}
            0,& X_{i0} < X_e,\\
            1,&   X_{i0} \geq X_e,\;t_{eq,i}> t_1 - \log\left(X_e/X_{i0}\right)/\mu_i,\\
            \dfrac{\log\left(X_{i0}/X_e\right)/\mu_i + \log\left(K/X_{l0}\right)/\mu_l - t_{eq,l}}{t_{eq,i} - t_{eq,l}},& \text{i.o.c.}
        \end{cases}
        \end{align*}      
    \end{enumerate}
\end{corollary}
\begin{proof}
\textcolor{white}{b}\\
\vspace{-0.3cm}
\begin{enumerate}
    \item[1)] The demonstration is analogous to Cor.~\ref{cor:ext_cyc_totalext} by noting that, from Prop.~\ref{prop:SC} and Cor.~\ref{cor:t_c=t_g}.2, any strain $S_l$ surviving cyclic treatment verifies:
    \begin{align*}
     t_{eq,l}\underbrace{=}_{\text{Prop.~\ref{prop:SC}}}\min\{t_{eq,j}\mid j\in J\} \underbrace{\leq}_{\text{Cor.~\ref{cor:t_c=t_g}.2}} t_c \leq t_g,\quad c> 1.
    \end{align*}
    \item[2)]
        Let us assume that strain $S_i$ $(1\leq i\leq n)$ goes extinct. Since the total population does not become extinct, we can take the following set formed by the indexes of the strains surviving cyclic treatment:
        \begin{align*}
           \Lambda = \{l\in J\mid  t_{eq,l}\leq t_{eq,j},\; j\in J\}\neq \varnothing,  
        \end{align*}
        and $i\notin \Lambda,$ by Prop.~\ref{prop:SC}. Note that the demonstration is trivial for  $c_{e,i}=0$ and $c_{e,i}=1,$  as observed in Cor.~\ref{cor:ext_cyc_totalext}. Conversely, if the extinction cycle of strain $S_i$ verifies $c_{e,i}>1,$ we can use Lem.~\ref{lemma:ext_Xg} to obtain:
        \begin{align*}
            X_i(t_{c_{e,i}}^f) < X_e &\Leftrightarrow X_i(t_{c_{e,i}}^f) \underbrace{=}_{\text{Eq.~\eqref{eq:Xf}}}D^{c_{e,i}}X_{i0}\exp\left(\mu_i\sum_{c=1}^{c_{e,i}} t_c - c_{e,i} k_i(C) t_k\right)<X_e \Leftrightarrow\\
        &c_{e,i}\left(\log(D) - k_i(C)t_k\right) + \mu_i\sum_{c=1}^{c_{e,i}} t_c < \log\left(\frac{X_e}{X_{i0}}\right)\Leftrightarrow  c_{e,i}t_{eq,i} - \sum_{c=1}^{c_{e,i}} t_c  > \frac{1}{\mu_i} \log\left(\frac{X_{i0}}{X_{e}}\right).
    \end{align*}

    However, we lack a closed-form expression linking the extinction cycle with the summation:
    \begin{align*}
        \tau_{c_{e,i}} = \sum_{c=1}^{c_{e,i}} t_c,
    \end{align*}
    which represents the time spent by the population in the exponential growth phase to the end of the extinction cycle $c_{e,i}>1.$ Then, we cannot find the exact value of the extinction cycle for strain $S_i,$ but we can take into account that any strain $S_l$ surviving cyclic treatment verifies:
    \begin{align*}
     l\in\Lambda \Rightarrow&\;    X_l(t_{c_{e,i}}^f) \underbrace{=}_{\text{Eq.~\eqref{eq:Xf}}} X_{l0}D^{c_{e,i}}\exp\left(\mu_l\sum_{c=1}^{c_{e,i}} t_c - c_{e,i}k_l(C)t_k\right) \underbrace{\leq}_{\text{Eq.~\eqref{eq:t_sat}}} DK\exp\left(-k_l(C)t_k\right)\Leftrightarrow\\
     &\; \sum_{c=1}^{c_{e,i}} t_c - c_{e,i} t_{eq,l}\leq \frac{1}{\mu_l} \log\left(\frac{K}{X_{l0}}\right) - t_{eq,l}\Leftrightarrow \sum_{c=1}^{c_{e,i}} t_c \leq \frac{1}{\mu_l} \log\left(\frac{K}{X_{l0}}\right) + (c_{e,i}-1) t_{eq,l},
    \end{align*}
    so that we can substitute the unknown bound for $c_{e,i}$ by the following:
    \begin{align*}
        c_{e,i}t_{eq,i} - \sum_{c=1}^{c_{e,i}} t_c  > \frac{1}{\mu_i} \log\left(\frac{X_{i0}}{X_{e}}\right),\quad \sum_{c=1}^{c_{e,i}} t_c \leq \frac{1}{\mu_l} \log\left(\frac{K}{X_{l0}}\right) + (c_{e,i}-1) t_{eq,l},\quad l\in\Lambda \Rightarrow \\
        c_{e,i}t_{eq,i} - \frac{1}{\mu_l} \log\left(\frac{K}{X_{l0}}\right) - (c_{e,i}-1) t_{eq,l} > \frac{1}{\mu_i} \log\left(\frac{X_{i0}}{X_e}\right),\quad l\in\Lambda\Rightarrow\\
        c_{e,i}\left(t_{eq,i} - t_{eq,l}\right) > \frac{1}{\mu_i} \log\left(\frac{X_{i0}}{X_e}\right) + \frac{1}{\mu_l} \log\left(\frac{K}{X_{l0}}\right) - t_{eq,l},\quad l\in\Lambda\Rightarrow \\[0.2cm]
        c_{e,i}  > L_{e,i}(l)=\frac{\log\left(X_{i0}/X_e\right)/\mu_i + \log\left(K/X_{l0}\right)/\mu_l - t_{eq,l}}{t_{eq,i} - t_{eq,l}},\quad l\in\Lambda.
    \end{align*}

\end{enumerate}

Note the difference between the bound in the extinction cycle for the case where the total population becomes extinct and the case where some strain is selected. When the total population becomes extinct after the first cycle, Cor.~\ref{cor:ext_cyc_totalext} states that the extinction cycle for any strain depends only on its bacterial traits and setup parameters (i.e., inoculum, growth and killing time, and extinction limit). Then, the extinction dynamics of the different strains from the first treatment cycle onwards are independent of each other. Nevertheless, when some strains survive the cyclic protocol, bacterial competition influences the extinction cycle of the strains which are going extinct, whether in isolation or not. That is, the extinction cycle for the strains going extinct also depends on the bacterial traits of the strains being selected and on their inocula, as shown in Cor.~\ref{cor:ext_cycle}. As a consequence, the extinction cycle can be seen as an indirect measure of the selective strength between competing strains.

\end{proof}


\section{Conclusions}\label{sec:Conclusions}

In this work, we introduce a mathematical model for cyclic antimicrobial treatment in bacterial populations formed by multiple strains. The treatment alternates growth periods, during which the strains compete indirectly for the available resources, and antimicrobial killing periods via some bactericidal substance. The model thus considers the growth and kill rates for the different strains as the fundamental parameters (bacterial traits) characterising the bacterial dynamics under cyclic treatment. The growth rate determines the ability of the strains to compete for resources during the growth periods, so that the fastest-growing strains reproduce more rapidly and consume all resources before slower-growing strains can convert them into biomass. On the other hand, the kill rate determines the  ability of the strains to survive the killing periods and compete in the next growth round. The model formulation was kept simple, minimising the number of bacterial traits considered, to analyse the effects of the most basic parameters (growth and kill rates) on the selection and extinction dynamics of the different strains during cyclic treatment.

Next, we used the proposed model to study the selection and extinction dynamics of the population as a function of the bacterial traits and the setup parameters for the cyclic protocol. In the first place, we derived necessary and sufficient conditions for determining the success of cyclic treatment in eliminating the bacterial population, so that none of the strains is selected. Interestingly, we found that the population can be eliminated with just one treatment cycle due to bacterial competition, even if the strains survive in isolation. Thus, mutant strains with reduced antimicrobial susceptibility, which are typically present in the population at very low initial frequencies, cannot be selected if competition with the ancestor (much more abundant initially) prevents the mutant from propagating beyond the first cycle. 

Once the condition for total population extinction was established, we used the mathematical model to derive selection conditions for the different strains when the cyclic antimicrobial protocol fails to eliminate the population. It was found that bacterial selection depends on the time required for the strains to reach equilibrium between new bacteria formed during growth and cell removal after antimicrobial killing and dilution. Then, the strains that are selected under cyclic treatment are those that minimise the equilibrium time, thereby achieving a better compromise between bacterial growth and death. This result is especially useful for determining the fixation of mutant strains over the ancestor when they compete under cyclic treatment, since mutations that reduce antimicrobial susceptibility usually entail a fitness cost (reduced growth) directly proportional to the reduction in antimicrobial effect. In this regard, we showed that the selection conditions can be equivalently expressed as a trade-off between fitness cost and increased antimicrobial survival conferred by the mutation, so that the mutant outcompetes the ancestor when the survival advantage during antimicrobial killing compensates for the fitness cost during growth.

The work developed here can be used from a theoretical perspective, allowing the study of the bacterial selection dynamics in antimicrobial protocols. The given approach is general enough to cover a wide range of selection scenarios involving antimicrobial resistance or tolerance, for example, and considering biocides with different antimicrobial properties (such as antibiotics or disinfectants). Additionally, the developed approach can be used to assist real-life antimicrobial protocols involving resistant or tolerant strains. Indeed, the selection and extinction conditions derived from the mathematical model can be applied directly to predict the success of cyclic treatment or to design optimal protocols, and the expression for the extinction cycles helps to determine the number of cycles necessary for the antimicrobial protocol to eliminate the different strains (if possible).

In future work, we hope the proposed approach  will serve as a basis for studying the influence of other bacterial characteristics that can be of importance for selection under cyclic treatment. Lag time or growth efficiency, for example, are bacterial traits that can vary significantly between strains and thus have a great impact on competition dynamics depending on the case. On the other hand, the different bacterial strains that form the population can exhibit significant phenotypic variability, leading to heterogeneous antimicrobial killing. For example, heteroresistance and persistence~\cite{2012Cogan,2019Balaban} translate into distinct reversible states (phenotypes) within the same (isogenic) strain, which differ in their antimicrobial susceptibility and other biological characteristics. Then, the transfer rate between different phenotypes and the intra-strain variability in other bacterial characteristics are parameters that can determine selection and seriously impact the outcome of antimicrobial protocols.

\bibliographystyle{unsrt}  
\bibliography{References_Arkiv}

\end{document}